\documentclass[conference]{IEEEtran}
\IEEEoverridecommandlockouts
\usepackage{cite}
\usepackage{amsmath,amssymb,amsfonts,amsthm}
\usepackage{graphicx}
\usepackage{soul}
\usepackage{textcomp}
\usepackage{nicefrac,xfrac}
\usepackage{flushend}
\usepackage{empheq}
\usepackage{xcolor}
\usepackage{comment}
\usepackage{hyperref}
\usepackage{algorithm}
\usepackage{algpseudocode}
\newtheorem{theorem}{Theorem}
\newtheorem{lemma}{Lemma}
\usepackage{pgfplots}

\theoremstyle{remark}
\newtheorem{rem}{\bf Remark}

\DeclareMathOperator*{\argmin}{argmin}
\usepackage{bbm}
\usepackage{circuitikz}
\usetikzlibrary{arrows,shapes}
\usepackage{tkz-euclide}
\usetikzlibrary{backgrounds}
\usepackage{mathtools}

\definecolor{azure}{rgb}{0.0, 0.5, 1.0}

\definecolor{chestnut}{rgb}{0.8, 0.36, 0.36}
\definecolor{airforceblue}{rgb}{0.36, 0.54, 0.66}
\definecolor{cadmiumorange}{rgb}{0.93, 0.53, 0.18}
\definecolor{bleudefrance}{rgb}{0.19, 0.55, 0.91}
\definecolor{carolinablue}{rgb}{0.6, 0.73, 0.89}
\definecolor{blue(ncs)}{rgb}{0.0, 0.53, 0.74}
\definecolor{dodgerblue}{rgb}{0.12, 0.56, 1.0}
\definecolor{cssgreen}{rgb}{0.0, 0.5, 0.0}
\definecolor{cadmiumgreen}{rgb}{0.0, 0.42, 0.24}
\definecolor{cadmiumorange}{rgb}{0.93, 0.53, 0.18}
\definecolor{amaranth}{rgb}{0.9, 0.17, 0.31}
\definecolor{bluegray}{rgb}{0.4, 0.6, 0.8}
\definecolor{cadmiumgreen}{rgb}{0.0, 0.42, 0.24}
\definecolor{chestnut}{rgb}{0.8, 0.36, 0.36}
\definecolor{ceil}{rgb}{0.57, 0.63, 0.81}
\definecolor{ashgrey}{rgb}{0.7, 0.75, 0.71}

\def\BibTeX{{\rm B\kern-.05em{\sc i\kern-.025em b}\kern-.08em
    T\kern-.1667em\lower.7ex\hbox{E}\kern-.125emX}}
    \bstctlcite{IEEEexample:BSTcontrol}
\begin{document}
\title{Multiple Receiver Over-the-Air Computation for Wireless Networked Control Systems\\
}

\author{Seif Hussein\textsuperscript{*}, Chinwendu Enyioha\textsuperscript{\textdagger},  Carlo Fischione\textsuperscript{*}
    \\
\IEEEauthorblockA{\textsuperscript{*}Division of Network and Systems Engineering, KTH Royal Institute of Technology, Stockholm, Sweden\\
\textsuperscript{\textdagger}Department of Electrical and Computer Engineering, University of Central Florida, Orlando, FL, USA\\
Email: \{seifh, carlofi\}@kth.se, cenyioha@ucf.edu}
}

\maketitle

\begin{abstract}
We propose a multi-sender, multi-receiver over-the-air computation (OAC) framework for wireless networked control systems (WNCS) with structural constraints. Our approach enables actuators to directly compute and apply control signals from sensor measurements, eliminating the need for a centralized controller. We use an iterative and convexifying procedure to obtain a control law that is structured with respect to the network topology and minimizes the overall system energy-to-energy gain. Furthermore, we solve a constrained matrix factorization problem to find the optimal OAC configuration with respect to power consumption, robustness, and stability of the WNCS. We prove the convergence of our proposed algorithms and present numerical results that validate our approach to preserve closed-loop stability with robust control performance and constrained power.

\end{abstract}


\section{Introduction} \label{sec:introduction}
Wireless control networks continue to be central to the development of distributed systems, with applications ranging from industrial automation and robotics to autonomous vehicles and Internet of Things devices\cite{pajic2011wireless,park2017wireless,zeng2019joint}. These systems rely on robust communication protocols to maintain stability and performance. Traditional wireless control methods often face challenges in terms of bandwidth limitations, latency, and signal reliability, particularly in large-scale networks with limited communication resources \cite{park2017wireless}.

Recent studies have shown that over-the-air computation (OAC) leverages the signal superposition property of multiple-access channels to enable faster and more communication-efficient signal aggregation during transmission \cite{csahin2023survey}. OAC emerges as a viable way to address the limitations of conventional wireless control systems by providing an architecture for actuators to process signals without need for explicit communication from each device to a centralized controller, thus reducing communication overhead and improving energy efficiency~\cite{park2021wireless}.

In wireless control systems, a primary objective is ensuring the stability of the overall network. In this context, stability refers to the ability of the system to maintain predictable and reliable performance despite external disturbances and fluctuations in the operating conditions of the network. The integration of OAC into wireless control systems adds a layer of complexity, as it introduces additional communication delays, noise, and possible interference that need to be taken into account when analyzing the closed-loop stability of the system.


An overview of existing literature that puts this work in context is presented in Section~\ref{sec:literature-review}. We present our system model and formally state our problem in Section~\ref{sec:problem}. In Section~\ref{sec:mainresults} our proposed algorithms are presented along with our main results. Numerical simulation results are presented in Section~\ref{sec:simulations} and the paper is concluded in Section~\ref{sec:conclude}.

\section{Related 
Work}\label{sec:literature-review}

Stability of WNCSs has been widely studied, with early studies focusing on traditional communication protocols that are inherently bandwidth-constrained. These systems often rely on periodic updates from sensors and controllers, leading to potential issues in terms of latency and the reliability of control signals, particularly in large-scale systems. 

In~\cite{abari2016over}, OAC was introduced as a method to reduce the communication burden in wireless sensor networks. OAC allows computation to be performed during the signal transmission process, rather than requiring a dedicated transmission phase, significantly reducing latency and increasing system bandwidth efficiency.
This makes it especially advantageous for large-scale networks with limited communication resources. The method has been widely studied in the distributed learning literature.
 Recently, early studies have explored the integration of OAC in network control where the control signals are computed at the transmitter, enabling efficiency of the overall network architecture. Stability continues to be the main objective, which is well understood in the traditional wireless control contexts~\cite{walsh2002stability,carnevale2007further} where Lyapunov-based approaches quantify the effect of the network channel conditions on the system stability. It is known that management of factors such as transmission power and signal-to-noise ratios (SNR) can impact stability. 

To leverage the benefits of OAC in wireless network control systems, some challenges are presented specifically with respect to achieving stability of closed-loop system. This is due to the dynamic conditions of the wireless channels, signal interference and noise. In fact, the impact of channel noise has been studied in \cite{csdesign}, where it was noted that noise-induced errors during signal transmission can accumulate and destabilize the system in not properly accounted for in the system design. The authors proposed adaptive techniques including power control and error correction methods to counteract the impact of noise on the system.
Another key factor that affects the closed-loop system stability in the multi-sender-multi-receiver OAC model considered here is the network topology. It is well known that the network topology plays an important role in determining the stability of a wireless control system~\cite{walsh2002stability}. However, the interaction between multiple sender-receiver pairs with the OAC architecture creates additional complex dynamics. This is especially true in the model considered in this paper, where the network topology has structured constraints. In \cite{10124016}, the authors address the multiple receiver OAC framework with a system model based on approximative matrix factorization. However, the approximative approach taken makes the framework unsuitable for WNCSs as the stability of the control system cannot be guaranteed.

Aggregating signals with OAC has been highlighted as a viable approach to increase the efficiency of WNCSs~\cite{park2021wireless,parkopt}. These approaches rely on a many-to-one architecture, which severely constrains the possible topologies of the network. In particular, the frameworks are only compatible with single-actuator systems. The OAC control framework has been extended to multiple actuators in ~\cite{10189841}, which proposes a multi-receiver approach. Although it is able to handle multiple actuators, this approach still relies on a \textit{central} controller which is \textit{wired} to all actuators, preventing fully distributed wireless network topologies.
\subsection{Our Contribution}
We present a time-slot-based many-to-many OAC architecture for a WNCS. Actuators receive information from sensor nodes directly by OAC, bypassing the need for a controller unit for aggregation and post-processing. Based on the WNCS topology, we formulate and solve a structured non-convex robust optimal control problem using a linearization-based algorithm. To obtain the optimal OAC configuration and ensure the stability of the WNCS, we formulate and solve a constrained non-convex matrix factorization problem using a modified Alternating Direction Method of Multipliers (ADMM) algorithm. The convergence of both algorithms to stationary points of the non-convex problems is established. Simulations validate the stable and robust performance of our approach.


\section{System Model}\label{sec:problem}
\subsection{Control System}
We consider a system with dynamics described by  \begin{equation}
    \label{eq:plant}
    \begin{aligned}
    \mathbf{x}[k+1] &= \mathbf{A}\mathbf{x}[k] + \mathbf{B}\mathbf{u}[k]\\
    \mathbf{y}[k] &= \mathbf{C}\mathbf{x}[k],
    \end{aligned}
\end{equation}
where $\mathbf{A} \in \mathbb{R}^{n \times n}, \mathbf{B}\in \mathbb{R}^{n \times m}$ and $\mathbf{C} \in \mathbb{R}^{p \times n}$.
Here, the output vector $\mathbf{y}[k] = [y_1[k] \: y_2[k] \ldots y_p[k]]^{\rm T}$ represents sensor measurements of the $n$-dimensional state vector $\mathbf{x}[k]$, where the sensors are denoted by $s_1, s_2, \ldots, s_p$. The input vector $\mathbf{u}[k] = [u_1[k] \: u_2[k] \ldots u_m[k]]^{\rm T}$ comprises signals delivered to the plant by actuators $a_1, a_2, \ldots, a_m$. Sensors communicate their signals to actuator nodes through a wireless network.

This wireless network is modeled by the graph $\mathcal{G} = \{\mathcal{V}, \mathcal{E}\}$, where $\mathcal{V} = \mathcal{V}_A \cup \mathcal{V}_S = \{a_1, \ldots, a_m, s_1, \ldots, s_p\}$ is the set of $m+p$ sensors and actuators. The edge set $\mathcal{E}$ characterizes the network's communication structure; specifically, an edge $(a_j, s_i) \in \mathcal{E}$ indicates that actuator $a_i$ is capable of receiving data directly from sensor $a_j$.

Furthermore, the input applied to the plant, $u_i[k]$, is defined as a linear combination of the values from nodes in the neighborhood of actuator $a_i$, given by 
\begin{equation}
    \label{eq:zik}
    u_i[k] = \sum_{s_j \in \mathcal{N}_{a_i}}g_{ij}y_j[k].
\end{equation}
The coefficients $g_{ij}$, which are nonzero scalars, determine the specific linear combinations computed by each actuator in the network.
\definecolor{sage}{HTML}{57d491}
\definecolor{niceblue}{HTML}{57a8d4}

\begin{figure}[t!]
\centering
\begin{tikzpicture}[transform shape, scale=1,thick]
\tikzstyle{every node}=[font=\footnotesize]
    
    

    \foreach \x/\y in {1.75/1,0.65/2,-0.95/p}
        \node at (5,\x-0.15) {$s_\y$};
    \node at (5,-0.3) {$ \vdots$};

    \filldraw[blue] (5.4,1.75-0.15) circle (1.5pt);
    \draw[->,blue] (5.4,1.75-0.15) -- (6.6,1.75-0.15);
    \draw[->,blue] (5.4,1.75-0.15) -- (6.6,0.65-0.15);
    \draw[->,blue] (5.4,1.75-0.15) -- (6.6,-0.95-0.15);

    \draw[->,dashed] (6.8,1.2-0.15) .. controls (6.2,1.2-0.15) .. (6,2-0.15);
    \node[] at (6,2.1) {\scriptsize$\mathcal{N}_{a_1}$};
    
    \filldraw[cyan] (5.4,0.65-0.15) circle (1.5pt);
    \draw[->,cyan] (5.4,0.65-0.15) -- (6.6,0.65-0.15);
    \draw[->,cyan] (5.4,0.65-0.15) -- (6.6,1.75-0.15);
    \draw[->,cyan] (5.4,0.65-0.15) -- (6.6,-0.95-0.15);

    \filldraw[magenta] (5.4,-0.95-0.15) circle (1.5pt);
    \draw[->,magenta] (5.4,-0.95-0.15) -- (6.6,-0.95-0.15);
    \draw[->,magenta] (5.4,-0.95-0.15) -- (6.6,0.65-0.15);
    \draw[->,magenta] (5.4,-0.95-0.15) -- (6.6,1.75-0.15);
    \draw[->,dashed] (6.8,-0.5-0.15) .. controls (6.2,-0.5-0.15) .. (6,-1.3-0.15);
    \node[] at (6,-1.7) {\scriptsize$\mathcal{N}_{a_m}$};

    \foreach \x/\y in {1.75/1,0.65/2,-0.95/m}
        \draw[fill=sage,rounded corners] (6.6,\x-0.15-0.3) rectangle node[black] {$a_\y$}++(0.5,0.5);
    \node at (6.9,-0.3) {$ \vdots$};

    
    
\end{tikzpicture}
\caption{Illustration of the WNCS topology in our setting. Each actuator $a_i$ has a neighborhood $\mathcal{N}_{a_i}$ defined by the connected sensor nodes $s_j$. As an example, we have that $\mathcal{N}_{a_1} = \{s_1, s_2,s_p\}$ in the figure.}
\label{fig:network}
\end{figure}
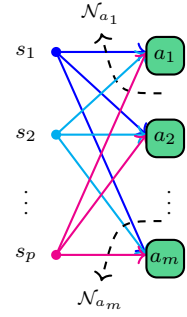

\subsection{Incorporating Over-the-Air Computation}
\label{sec:ota}
We consider an OAC system which consists of multiple receivers and multiple senders, where actuators and sensors are receivers and senders, respectively.  
%
Given the many-to-many information aggregation architecture, we consider multiple time slots $t \in [T]$ for each transmission similar to~\cite{10124016}. Specifically, a sensor $s_j$ at slot $t$ sends
\begin{equation}
    \label{eq:sender}
    p_{jt}y_j[k] = p_{jt}\mathbf{c}_j^{\rm T}\mathbf{x}[k],
\end{equation}
where $\mathbf{c}_j$ is row $j$ of the system output matrix $\mathbf{C}$ and $p_{jt}$ is the precoder of sensor $s_j$ at time slot $t$.
An actuator $a_i$ at slot $t$ receives
\begin{equation}
    r_{it}[k] = \sum_{s_j \in \mathcal{N}_{a_i}}p_{jt}h_{ij}y_j[k] + n_{it}[k],
\end{equation}
where $h_{ij}$ is the channel gain coefficient between sensor $s_j$ and actuator $a_i$. Here $n_{it}[k]$ is the Gaussian noise of the received signal assumed to have a bounded variance, and $h_{ij}$ is assumed to be known or estimated in this setting~\cite{9095231}. Furthermore, we assume that the total transmission time for $T$ slots is less than the sampling period $\delta$ of the control system.

Following the $T$ sets of transmissions, an estimate of the received signal $r_{it}[k]$ is obtained by actuator $a_i$ through decoding as
\begin{equation}
    \label{eq:rhohat}
    \hat{\rho}_i[k] = \sum_{t \in [T]}d_{it}r_{it}[k],
\end{equation}
where $d_{it}$ corresponds to the decoding coefficient of $a_i$ at time slot $t$.
The goal is to estimate the update procedures in \eqref{eq:sender}, namely
\begin{equation}
    \label{eq:rho}
    \rho_i[k] = 
    \sum_{s_j \in \mathcal{N}_{a_i}} g_{ij}y_{j}[k],
\end{equation}
which requires the precoding and decoding procedure to satisfy
\begin{equation}
    \label{eq:unbias}
    \sum_{t \in [T]}p_{jt}d_{it} = \frac{1}{h_{ij}}g_{ij},
\end{equation}
for all sensor-actuator pairs $(a_i, s_j)$.
Additionally, we impose a power constraint $P_j$ on the symbols at each slot $t$ for the sensor $s_j$ as
\begin{equation}
    \label{eq:powconst}
    |p_{jt}|^2 \leq P_j, \: t\in[T].
\end{equation}
Without loss of generality, we can normalize the state impact on the power constraint, absorbing $y_j[k]$ from \eqref{eq:sender} into $P_j$.
\subsection{The Closed Loop System}
The update \eqref{eq:sender} following the OAC procedure becomes
\begin{equation}
    u_{i}[k] = \sum_{s_j \in \mathcal{N}_{a_i}}g_{ij}y_j[k] + \sum_{t\in[T]}d_{it}n_{it}[k],
\end{equation}
and aggregating the values of all actuators at time step $k$ into the vector $\mathbf{u}[k] = [u_1[k] \: u_2[k] \ldots u_m[k]]^{\rm T}$, yields
\begin{equation}
    \mathbf{u}[k] = \mathbf{G}\mathbf{y}[k] + \sum_{t\in [T]}\mathbf{D}_t\mathbf{n}_t[k],
\end{equation}
where $\mathbf{G} \in \mathbb{R}^{m \times p}$ and $\mathbf{D}_t \in \mathbb{R}^{m \times m}$ is a diagonal matrix comprising the decoding coefficients $d_{it}$ in entry $(i,i)$.
In the above equations, for all $a_i \in \mathcal{V}$, we have $g_{ij} = 0$ if $s_j \notin \mathcal{N}_{a_i}$. Therefore, the matrix $\mathbf{G}$ is structured with respect to the network topology.
The closed loop system in \eqref{eq:plant} becomes
\begin{equation}
    \label{eq:closed loop}
    \mathbf{x}[k+1] = \hat{\mathbf{A}}\mathbf{x}[k] + \mathbf{B}\hat{\mathbf{n}}[k],
\end{equation}
where
\begin{equation}
    \hat{\mathbf{A}} = (\mathbf{A}+\mathbf{B}\mathbf{G}\mathbf{C}), \:
    \hat{\mathbf{n}}[k] = \sum_{t \in [T]}\mathbf{D}_t\mathbf{n}_t[k].
\end{equation}
\vspace{-1.5pt}
In the following section, we address the stabilization of \eqref{eq:closed loop} under the OAC approach that we propose.

\section{Structured Robust Control and OAC by Matrix Factorization}
\label{sec:mainresults}
\subsection{Performance and Stability of the Closed Loop System}
In this section, we look at the stabilization of the closed-loop system \eqref{eq:closed loop} under a suitable performance metric that quantifies the system's response to disturbances, in this case $\hat{\mathbf{n}}[k]$. A commonly used metric is the energy-to-energy gain $\gamma_{ee}$~\cite{unifiedalgebraic}, which in our case takes the form,
\begin{equation}
    \gamma_{ee} = \sup_{\|\hat{\mathbf{n}}\|_{\ell_2}\leq 1}\ \|\mathbf{x}\|_{\ell_2},
\end{equation}
where $\|\mathbf{x}\|_{\ell_2} = \sqrt{\sum_{k=0}^{\infty}\|\mathbf{x}[k]\|_2^2}$.
This metric characterizes the worst-case amplification of the disturbance $\hat{\mathbf{n}}[k]$ through the system. This is a suitable metric since the decoding coefficients are not known \textit{a priori} and affect the impact of the noise on the control system.

To construct a network that minimizes $\gamma_{ee}$ for \eqref{eq:closed loop}, we use the following established result:
\begin{lemma}~\cite{unifiedalgebraic}
    \label{thm:1}
    For a stable system \eqref{eq:closed loop}, and scalar $\gamma > 0$, the energy-to-energy gain satisfies
    $\gamma_{ee} < \gamma$ if and only if there exists a symmetric, positive-definite $\mathcal{X}$ such that
    \begin{equation}
        \label{eq:thm1ineq}
        \begin{bmatrix}
            \mathcal{X} & \mathbf{0}\\
            \mathbf{0} & \gamma^2 \mathbf{I}
        \end{bmatrix} \succ \begin{bmatrix}\hat{\mathbf{A}} &\mathbf{B}\\
                            \mathbf{I} & \mathbf{0}
        \end{bmatrix}
        \begin{bmatrix}
            \mathcal{X} & \mathbf{0}\\
            \mathbf{0} & \mathbf{I} 
        \end{bmatrix}
        \begin{bmatrix}
            \hat{\mathbf{A}}^{\rm T} & \mathbf{I}\\
           \mathbf{B}^{\rm T} & \mathbf{0}
        \end{bmatrix}.
    \end{equation}
\end{lemma}
Since the matrix $\mathbf{G}$ is structured with respect to the network topology, we constrain $\mathbf{G}$ to a fixed sparsity pattern. However, we can consider any convex constraint set $\mathcal{D}_{\mathbf{G}}$ for the matrix $\mathbf{G}$. Using Schur-complements on the condition in \eqref{eq:thm1ineq}, we seek to solve the optimization problem
\begin{equation}
    \label{eq:hinf}
    \begin{aligned}
        &\min_{\gamma, \mathcal{X}, \mathbf{G}} \quad \gamma^2\\
        &\text{s.t.}    \begin{bmatrix}
        \mathcal{X} & \mathbf{0} & \hat{\mathbf{A}} & \hat{\mathbf{B}}\\
        \mathbf{0} & \gamma^2 \mathbf{I} & \mathbf{I} & \mathbf{0}\\
        \hat{\mathbf{A}}^{\rm T} & \mathbf{I} & \mathcal{X}^{-1} & \mathbf{0}\\
        \hat{\mathbf{B}}^{\rm T} & \mathbf{0} & \mathbf{0} & \mathbf{I}
    \end{bmatrix} \succ \mathbf{0}, \: \mathcal{X} \succ \mathbf{0},\\
    \vspace{5pt}
    & \mathbf{G} \in \mathcal{D}_{\mathbf{G}},
    \end{aligned}
\end{equation}
where we suppress the dependency $\hat{\mathbf{A}} = \hat{\mathbf{A}}(\mathbf{G})$ for compactness. Problem \eqref{eq:hinf} is \textit{non-convex}, and a common approach to address the non-convexity is to linearize the problematic entry $\mathcal{X}^{-1}$~\cite{1272453,914229,618250}. One can then use the following iterative procedure to solve the resulting linearized formulation of \eqref{eq:hinf},
\begin{equation}
    \label{eq:relaxedh}
    \begin{aligned}
    \mathcal{X}_{k+1} &= \argmin_{\gamma, \mathcal{X}, \mathbf{G}} \quad \gamma^2\\
    &\text{s.t.}
    \begin{bmatrix}
        \mathcal{X} & \mathbf{0} & \hat{\mathbf{A}} & \hat{\mathbf{B}}\\
        \mathbf{0} & \gamma^2 \mathbf{I} & \mathbf{I} & \mathbf{0}\\
        \hat{\mathbf{A}}^{\rm T} & \mathbf{I} & L(\mathcal{X}^{-1},\mathcal{X}_k) & \mathbf{0}\\
        \hat{\mathbf{B}}^{\rm T} & \mathbf{0} & \mathbf{0} & \mathbf{I}
    \end{bmatrix} \succ \mathbf{0}, \mathcal{X} \succ \mathbf{0},\\
    \vspace{5pt}
    & \mathbf{G}\in \mathcal{D}_{\mathbf{G}},
    \end{aligned}
    \end{equation}
where $L(\mathcal{X}^{-1},\mathcal{X}_k) = \mathcal{X}_k^{-1} - \mathcal{X}_k^{-1}(\mathcal{X}-\mathcal{X}_k)\mathcal{X}_k^{-1}$ is the linearization of $\mathcal{X}^{-1}$ at the point $\mathcal{X}_k$. Here \eqref{eq:relaxedh} constitutes a sequence of \textit{convex} problems, and the following theorem establishes the convergence of the iterates,
\begin{theorem}    
\label{theorem:1}
    Let $\{\gamma_{k}, \mathcal{X}_{k},\mathbf{G}_k\}_{k=0}^{\infty}$ be a sequence generated by \eqref{eq:relaxedh} from a feasible initial point. Then every accumulation point of the sequence is a stationary point of \eqref{eq:hinf}.
\end{theorem}
\begin{proof}
See Appendix \ref{appendix:1}.
\end{proof}
\begin{rem}
Problem \eqref{eq:relaxedh} is solved \textit{offline}. Techniques such as in~\cite{618250} can be used to initialize \eqref{eq:relaxedh}.
\end{rem}
\vspace{-5pt}
\subsection{Optimal Over-the-Air Computation}
After obtaining $\mathbf{G}$ from \eqref{eq:optprob}, we look for the optimal precoders and decoders that minimize the error between \eqref{eq:rho} and the estimate \eqref{eq:rhohat} subject to the constraints \eqref{eq:unbias} and \eqref{eq:powconst}. Formally, we seek to solve the following optimization problem,
\begin{equation}
    \label{eq:prob1}
    \begin{aligned}
    &\min \sum_{a_i \in \mathcal{V}}^{}\:(1/2)\mathbb{E}[|\hat{\rho}_i - \rho_i|^2]\\
    &\text{s.t. }\sum_{t \in [T]}p_{jt}d_{it} = \frac{g_{ij}}{h_{ij}}, \forall(a_i,s_j) \in \mathcal{E}\\
    &\quad \:\: \:|p_{jt}|^2 \leq P_j , \forall s_j, t \in [T].
    \end{aligned}
\end{equation}
The expected value can be equivalently represented as,
\begin{equation}
    \mathbb{E}[|\hat{\rho}_i - \rho_i|^2] = \sigma^2\sum\nolimits_{t\in[T]}|d_{it}|^2,
\end{equation}
for $n_{it}[k] \in \mathcal{N}(0,\sigma^2)$.
We define the following matrices of precoding and decoding vectors,
\begin{equation}
    \begin{aligned}
    \mathbf{P} &:= [\mathbf{p}_1, \ldots,\mathbf{p}_{p}] \in \mathbb{R}^{T\times p}\\
    \mathbf{D} &:= [\mathbf{d}_1, \ldots, \mathbf{d}_{m}] \in \mathbb{R}^{T\times m},
    \end{aligned}
\end{equation}
where the vectors $\mathbf{p}_{j}$ and $\mathbf{d}_{i}$ contain the $T$ precoding and decoding coefficients for connected sensor-actuator pairs $(s_j, a_i)$.

Using these definitions, problem \eqref{eq:prob1} can be expressed as a matrix factorization problem
\begin{equation}
    \label{eq:optprob}
    \begin{aligned}
    &\min_{\mathbf{P},\mathbf{D}} \:(1/2)\|\mathbf{D}\|_{\rm F}^2\\
    &\text{s.t. } (\mathbf{G} \odot \mathbf{H}^{-1})^{\rm T} = \mathbf{P}^{\rm T}\mathbf{D}\\
    &\quad \: \: \: \|\mathbf{p}_j\|_2^2 \leq P_j, \: j=1,\ldots,p,
    \end{aligned}
\end{equation}
where $(\mathbf{H})_{ij} = h_{ij}$ contains the channel coefficients and $\odot$ denotes the Hadamard product. \begin{rem}
    The matrix factorization problem in \eqref{eq:optprob} is \textit{always} feasible for $P_j >0, h_{ij} \neq 0$ and $T \geq \text{rank}(\mathbf{G})$.
\end{rem}
Despite the non-convexity of \eqref{eq:optprob} due to the equality constraint, the problem is bi-convex, which motivates an alternating optimization approach, as the subproblems in $\mathbf{P}$ and $\mathbf{D}$ are convex. Moreover, ADMM-based methods have been shown to perform well in nonconvex settings with exploitable structure \cite{sindri}, which motivates our use of a modified ADMM approach to handle the constraints and leverage the biconvexity of \eqref{eq:optprob}.
Specifically, we define a modified augmented Lagrangian $\hat{\mathcal{L}}_{\tau^{(k)}}$ corresponding to \eqref{eq:optprob} as
\begin{equation}
    \label{eq:modaug}
    \begin{aligned}
    \hat{\mathcal{L}}_{\tau^{(k)}}(\mathbf{D},\mathbf{P},\mathbf{\Lambda}) &:= \mathcal{L}_{\tau^{(k)}}(\mathbf{D},\mathbf{P},\mathbf{\Lambda}) + (\alpha/2)\|\mathbf{D}^{(k)}-\mathbf{D}\|_{\rm F}^2\\
    &+(\beta/2)\|\mathbf{P}^{(k)}-\mathbf{P}\|_{\rm F}^2,
    \end{aligned}
\end{equation}
where 
\begin{equation}
    \label{eq:aug}
    \begin{aligned}
    \mathcal{L}_{\tau^{(k)}}(\mathbf{D},\mathbf{P},\mathbf{\Lambda}) &:= (1/2)\|\mathbf{D}\|_{\rm F}^2 - \|\mathbf{\Lambda}\|_{\rm F}^2/(2\tau^{(k)}) + \mathbf{I}_{\Omega}(\mathbf{P})\\&+ (\tau^{(k)}/2)\|(\mathbf{G} \odot \mathbf{H}^{-1})^{\rm T} - \mathbf{P}^{\rm T}\mathbf{D} + \mathbf{\Lambda}/\tau^{(k)}\|_{\rm F}^2.
    \end{aligned}
\end{equation}
Here $\mathbf{\Lambda}$ is a dual variable, the indicator function $\mathbf{I}_{\Omega}(\mathbf{P})$ corresponds to the set $\Omega=\{\mathbf{P}:\|\mathbf{p}_j\|_2^2\leq P_j\}$ and $\tau^{(k)}$ is the step size at iteration $k$. Furthermore, we have introduced two proximal terms in \eqref{eq:modaug} with $\alpha,\beta > 0$. We summarize the steps for our modified ADMM-algorithm in Algorithm \ref{alg:admm}, with convergence and step-size rule given by the following theorem

\begin{theorem}

\label{theorem:2}
The iterates of Algorithm \ref{alg:admm} generate a sequence $\{\mathbf{D}^{(k)}, \mathbf{P}^{(k)}, \mathbf{\Lambda}^{(k)}\}_{k=0}^{\infty}$ that converges to a stationary point of \eqref{eq:optprob} if $\|\mathbf{\Lambda}^{(k)}\|_{\rm F} \leq M$ and $\tau^{(k+1)}/\tau^{(k)} \leq C$ for some constants $C,M$ and all $k$, and $\sum_{k=1}^\infty\frac{1}{\tau^{(k)}} < \infty$.
\end{theorem}
\begin{proof}
    See Appendix \ref{appendix:2}.
\end{proof}
We motivate the boundedness assumption of the Lagrange multiplier by the following Lemma
\begin{lemma}
    \label{lem:2}
    For a feasible problem \eqref{eq:optprob}, the stationary Lagrange multiplier $\mathbf{\Lambda}$ corresponding to the equality constraint is bounded if $T \geq \text{rank}(\mathbf{G})$.
\end{lemma}
\begin{proof}
See Appendix \ref{appendix:3}.
\end{proof}
\begin{rem}
For any feasible problem \eqref{eq:optprob}, the equality constraint guarantees that the reconstructed closed-loop matrix $\hat{\mathbf{A}} = \mathbf{B}(\mathbf{H}\odot\mathbf{P}^{\rm T}\mathbf{D})^{\rm T}\mathbf{C}$ is Schur since $(\mathbf{H}\odot\mathbf{P}^{\rm T}\mathbf{D})^{\rm T} = \mathbf{G}$, where $\mathbf{G}$ is from \eqref{eq:relaxedh}.
\end{rem}
\begin{algorithm}[t!]
\caption{Modified ADMM for \eqref{eq:optprob}}
\label{alg:admm}
\begin{algorithmic}
\State Input: $\mathbf{G}, \mathbf{H}, \tau^{(0)},\alpha,\beta$
\State $\mathbf{\Lambda}^{(0)}\gets \mathbf{0}$
\State Generate $\mathbf{P}^{(0)},\mathbf{D}^{(0)}$ randomly
\For {$k \leq \rm{T_{\max}}$}
 \State \hspace{-7pt}$\mathbf{D}^{(k+1)} \gets \argmin_{\mathbf{D}} \: \hat{\mathcal{L}}_{\tau^{(k)}}(\mathbf{D},\mathbf{P}^{(k)},\mathbf{\Lambda}^{(k)})$
 
 \State \hspace{-10pt} $\mathbf{P}^{(k+1)} \gets \argmin_{\mathbf{P}} \: \hat{\mathcal{L}}_{\tau^{(k)}}(\mathbf{D}^{(k+1)},\mathbf{P},\mathbf{\Lambda}^{(k)})$
 
 \State \hspace{-10pt} $\mathbf{\Lambda}^{(k+1)} \gets \mathbf{\Lambda}^{(k)} + \tau^{(k)}((\mathbf{G}\odot \mathbf{H}^{-1})^{\rm T} - \mathbf{P}^{(k+1)^{\rm T}}\mathbf{D}^{(k+1)})$
 
 \hspace{-11pt}Update $\tau^{(k)} < \tau^{(k+1)}$
\EndFor 
\end{algorithmic}
\end{algorithm}

\vspace{-6pt}
\section{Simulation Results}\label{sec:simulations}

In this section, we present two sets of numerical experiments to illustrate the advantages of our OAC control framework. Unless otherwise stated, we use $\alpha = \beta = 0.1$ and $\tau^{(k)} = k^{1.5}$ in Algorithm \ref{alg:admm}. The subproblem in the $\mathbf{P}$-update is solved using standard ADMM.

\subsection{Stability under Power Constraints}
The stability of system \eqref{eq:closed loop} is investigated for a sequence of power constraints $P_j = P_{\max} \in\{0.1, 0.2,\ldots,1\}$ for all $j$, in \eqref{eq:optprob}.
Matrices ($\mathbf{A}, \mathbf{B}$,$\mathbf{C}$) are randomly generated with $n=6$, $p=m=4$. For a given triplet $(\mathbf{A}, \mathbf{B}, \mathbf{C})$, we find the stabilizing $\mathbf{G}$ from \eqref{eq:relaxedh}. We compare our method \eqref{eq:optprob} of reconstructing $\mathbf{G} \approx \mathbf{H} \: \odot \:\mathbf{D}\mathbf{P}^{\rm T}$ to the unconstrained multiple-receiver OAC approach in~\cite{10124016}. The number of timeslots is set to $T = 4$ for both methods and the channel coefficients $h_{ij}$ are distributed as Rayleigh$(1)$. For the experiment, we generate $100$ triplets $(\mathbf{A},\mathbf{B},\mathbf{C})$ per power constraint $P_{\max}$ and compute the percentage of unstable matrices $\hat{\mathbf{A}}=\mathbf{A}+\mathbf{B}\mathbf{H} \odot \mathbf{D}\mathbf{P}^{\rm T}\mathbf{C}$.

Fig. \ref{fig:rhoA} shows the percentage of unstable closed-loop matrices, i.e., $\rho(\hat{\mathbf{A}}) > 1$, for our approach and the one in~\cite{10124016}. We note that our approach conserves stability due to the equality constraint in \eqref{eq:optprob}, while the unconstrained approach in~\cite{10124016} compromises the stability of the system. In particular, we can recover an exact factorization of $\mathbf{G}$ for the same number of timeslots $T$.
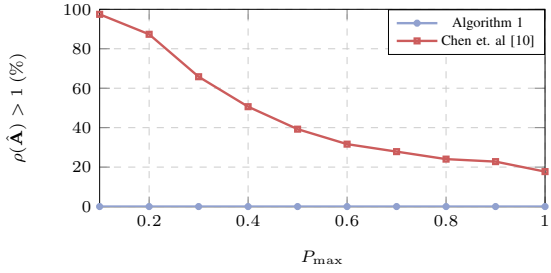
\begin{figure}[!t]
\centering
    \begin{tikzpicture}[transform shape,scale=0.9]
    \begin{axis}[
        xlabel={$P_{\max}$},
        ylabel={$\rho(\hat{\mathbf{A}}) > 1\: (\%)$},
        label style={font=\scriptsize},
        tick label style={font=\scriptsize} , 
        width=0.45\textwidth,
        height=4.5cm,
        xmin=0.1, xmax=1,
        ymin=-0.1, ymax=100,
       legend style={nodes={scale=0.6, transform shape}, at={(1,1)}}, 
        ymajorgrids=true,
        xmajorgrids=true,
        grid style=dashed,
        grid=both,
        grid style={line width=.1pt, draw=gray!15},
        major grid style={line width=.2pt,draw=gray!40},
    ]
     \addplot[
        color=ceil,
        mark=*,
        line width=1pt,
        mark size=1pt,
        ]
    table[x=x,y=C]
    {data/oureigs.dat};
     \addplot[
        color=chestnut,
        mark=square,
        line width=1pt,
        mark size=1pt,
        ]
    table[x=x,y=C]
    {data/zheng.dat};
    \legend{Algorithm \ref{alg:admm}, Chen et. al \cite{10124016}, legend pos = north east};
    \end{axis}
\end{tikzpicture}
  \caption{Percentage of unstable closed loop matrices $\hat{\mathbf{A}} = \mathbf{A}+\mathbf{B}\mathbf{G}\mathbf{C}$ with $\mathbf{G}$ reconstructed as $\mathbf{G} \approx (\mathbf{H}\odot\mathbf{P}^{\rm T}\mathbf{D})^{\rm T}$ from \eqref{eq:optprob}. The closed loop matrix $\hat{\mathbf{A}}$ is obtained for $100$ independent random initializations of $(\mathbf{A},\mathbf{B},\mathbf{C})$ in \eqref{eq:relaxedh}.}
  \vspace{-5pt}
  \label{fig:rhoA}
\end{figure}

\subsection{Control Peformance}
We consider the ball and beam system with $n=p = 4$, $m=1$ and discretization $\delta = 100$ ms. The noise variance is set to $\sigma^2 = 0.01$, $h_{ij} \in $ Rayleigh$(1)$ and the power constraint $P_j \in \{0.1,\ldots,1\}$ for all $j$. The SNR is $10\log (P_j/\sigma^2)$. Each data point is taken as the average MSE of 100 Monte Carlo simulations, each corresponding to a 5 second evolution of the state $\mathbf{x}[k]$.
We compare our approach with the single actuator $\mathcal{H}_\infty$ OAC approach in~\cite{parkopt}. For a fair comparison, we set $\mathbf{C} = \mathbf{I}$ since this is necessary in their system model. Moreover, we find that the algorithm proposed in~\cite{parkopt} is infeasible with respect to the power constraint for most of the realizations of $h_{ij}$. For the sake of fairness, we scale their corresponding decoding factors accordingly to compensate for the cases of infeasibility.

Fig. \ref{fig:MSEplot} shows the average MSE for the unconstrained problem, as well as for the cases where $\|\mathbf{G}\|_{\rm F} \leq 50$ and $\|\mathbf{G}\|_{\rm F} \leq ~36$ in \eqref{eq:relaxedh}. Compared to the OAC approach in~\cite{parkopt}, we observe significant reductions in MSE. Furthermore, constraining the entries of \(\mathbf{G}\) to be smaller results in a lower MSE due to the interaction between \(\mathbf{P}\) and \(\mathbf{D}\) in \eqref{eq:optprob}. In particular, having small-magnitude entries in \(\mathbf{G}\) allows the entries in \(\mathbf{D}\) to be smaller, leading to reduced amplification of \(\hat{\mathbf{n}}[k]\) in \eqref{eq:closed loop}.

\begin{figure}[!t]
\centering
    \begin{tikzpicture}[transform shape,scale=0.9]
    \begin{axis}[
        xlabel={SNR (dB)},
        ylabel={MSE},
        label style={font=\scriptsize},
        tick label style={font=\scriptsize} , 
        width=0.45\textwidth,
        height=5cm,
        xmin=10, xmax=20,
        ymin=1e-3, ymax=2e1,
         ymode = log,
       legend style={nodes={scale=0.45, transform shape}, at={(1,1)}}, 
        ymajorgrids=true,
        xmajorgrids=true,
        grid style=dashed,
        grid=both,
        grid style={line width=.1pt, draw=gray!15},
        major grid style={line width=.2pt,draw=gray!40},
    ]
     \addplot[
        color=cadmiumorange,
        mark=triangle,
        line width=1pt,
        mark size=1.5pt,
        ]
    table[x=x,y=C]
    {data/Park.dat};
     \addplot[
        color=ceil,
        mark=*,
        line width=1pt,
        mark size=1.5pt,
        ]
    table[x=x,y=C]
    {data/Unconstr.dat};
     \addplot[
        color=chestnut,
        mark=x,
        line width=1pt,
        mark size=1.5pt,
        ]
    table[x=x,y=C]
    {data/50.dat};
     \addplot[
        color=cssgreen,
        mark=square,
        line width=1pt,
        mark size=1.5pt,
        ]
    table[x=x,y=C]
    {data/36.dat};
    \legend{Park et. al \cite{parkopt}, Unconstrained, $\|\mathbf{G}\|_{\rm F} \leq 50$,$\|\mathbf{G}\|_{\rm F}\leq36$, legend pos = north east};
    \end{axis}
\end{tikzpicture}
\vspace{-7.5pt}
  \caption{Average MSE of state-vector $\mathbf{x}[k]$ over $5$ seconds in the ball and beam system with sampling period $\delta=100$ ms. Each data point is the average of 100 Monte Carlo simulations.}
  \vspace{-5pt}
  \label{fig:MSEplot}
\end{figure}
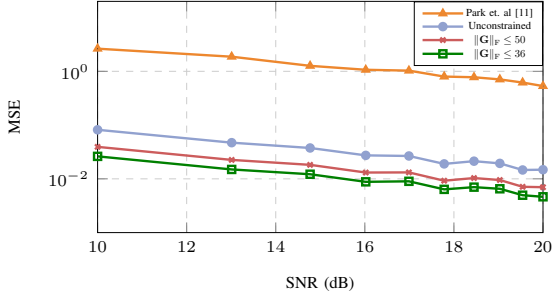



\section{Conclusion}\label{sec:conclude}
In this paper, we developed a multi-sender, multi-receiver over-the-air (OAC) framework for WNCSs, which allows actuators to directly compute and apply control signals from sensor measurements without requiring a centralized controller. We introduced an iterative linearization procedure for a stable control system subject to an energy-to-energy performance metric and a structured network topology.
By formulating the design of OAC transmissions as a constrained matrix factorization problem, we preserve the stability of the WNCS while accounting for channel noise, interference, and limited power budgets. Numerical results confirmed that the proposed method achieves robust performance and closed-loop stability under communication and power constraints, demonstrating significant advantages over existing single-receiver OAC schemes. 
 \bibliographystyle{IEEEtran}
 \bibliography{Ref-files}
\appendix
\subsection{Proof of Theorem 1}
The map $\mathcal{X} \to \mathcal{X}^{-1}$ is operator-convex on the cone of positive definite matrices~\cite{bhatia1997matrix}, meaning for $\mathcal{X} \succ 0,\mathcal{Y} \succ 0$ and $t\in(0,1)$, the following holds
\begin{equation}
    (t \mathcal{X} + (1-t)\mathcal{Y})^{-1} \preceq t \mathcal{X}^{-1} + (1-t)\mathcal{Y}^{-1}.
\end{equation}
Consequently, we have for $\mathcal{X}_k \succ 0, \mathcal{X}_k+\Delta \succ 0$,
\begin{equation}
    \begin{aligned}
    \frac{(\mathcal{X}_k + t\Delta)^{-1} - \mathcal{X}_k^{-1}}{t} &= \frac{((1-t)\mathcal{X}_k + t(\mathcal{X}_k + \Delta))^{-1} - \mathcal{X}_k^{-1}}{t}\\
    &\preceq \frac{(1-t)\mathcal{X}_k^{-1} + t(\mathcal{X}_k+\Delta)^{-1}-\mathcal{X}_k^{-1}}{t}\\
    &= -\mathcal{X}_k^{-1} + (\mathcal{X}_k+ \Delta)^{-1},
    \end{aligned}
\end{equation}
so by taking the Fréchet derivative, we get 
\begin{equation}
    \begin{aligned}
     \lim_{t\to0^{+}} \frac{(\mathcal{X}_k + t\Delta)^{-1} - \mathcal{X}_k^{-1}}{t}&=  -\mathcal{X}_k^{-1}\Delta\mathcal{X}_k^{-1} \\
     &\preceq -\mathcal{X}_k^{-1} + (\mathcal{X}_k+\Delta)^{-1}.
    \end{aligned}
\end{equation}
Now, let $\Delta = \mathcal{X} - \mathcal{X}_k$, which yields
\begin{equation}
    \label{eq:invmajor}
    -\mathcal{X}_k^{-1}(\mathcal{X}-\mathcal{X}_k)\mathcal{X}_k^{-1} + \mathcal{X}_k^{-1} \preceq \mathcal{X}^{-1},
\end{equation}
meaning $\mathcal{X}^{-1}$ majorizes its linearization at $\mathcal{X}_k$ if $\mathcal{X}, \mathcal{X}_k \succ 0$. Hence, it is straightforward to see that any feasible solution to \eqref{eq:relaxedh} must be feasible to \eqref{eq:hinf}.

From the definition of $L(\mathcal{X}^{-1},\mathcal{X}_k)$, it is clear that $\mathcal{X}_k$ is feasible with respect to \eqref{eq:relaxedh} at the subsequent iteration $k+1$. Consequently, we see that if $\mathcal{X}_{k+1} \neq \mathcal{X}_k$ then $\gamma_{k+1} < \gamma_k$, and 
since $\gamma$ is lower bounded in the feasible set, the sequence $\gamma_{k}$ decreases unless $\gamma_{k+1} = \gamma_{k}$.
\label{appendix:1}
\subsection{Proof of Theorem 2}
\label{appendix:2}
Define $\Delta \mathbf{X}^{(k+1)} := \mathbf{X}^{(k)}-\mathbf{X}^{(k+1)}, \delta \tau^{(k+1)}:=\frac{\tau^{(k)}+\tau^{(k+1)}}{2(\tau^{(k)})^2}$ as shorthands.
The following Lemma is used to establish a pseudo-descent for the augmented Lagrangian
\begin{lemma}
The following bound holds for $\mathbf{D}^{(k+1)}, \mathbf{P}^{(k+1)}$ and $\mathbf{\Lambda}^{(k+1)}$ generated from Algorithm \ref{alg:admm}:
\begin{equation}
    \begin{aligned}
    &\mathcal{L}_{\tau^{(k)}}(\mathbf{P}^{(k)},\mathbf{D}^{(k)},\mathbf{\Lambda}^{(k)}) - \mathcal{L}_{\tau^{(k+1)}}(\mathbf{P}^{(k+1)},\mathbf{D}^{(k+1)},\mathbf{\Lambda}^{(k+1)}) \\
    &\geq (\alpha/2)\|\Delta \mathbf{D}^{(k+1)}\|_{\rm F}^2
        + (\beta/2)\|\Delta \mathbf{P}^{(k+1)}\|_{\rm F}^2\\
        &- \delta \tau^{(k+1)}\|\Delta \mathbf{\Lambda}^{(k+1)}\|_{\rm F}^2.
    \end{aligned}
\end{equation}
\label{lemma:1}
\end{lemma}
\vspace{-25pt}
\begin{proof}
A similar proof can be found in~\cite{MatFact}, the derivation is analogous.
\end{proof}
We use the shorthand notation $\mathcal{L}_{\tau^{(k)}}^k := \mathcal{L}_{\tau^{(k)}}(\mathbf{P}^{(k)}, \mathbf{D}^{(k)}, \mathbf{\Lambda}^{(k)})$. Using Lemma \eqref{lemma:1}, we consider the summation of differences of successive iterations $\mathcal{L}_{\tau^{(k)}}^k - \mathcal{L}_{\tau^{(k+1)}}^{k+1}$,
\begin{equation}
    \label{eq:augineq}
    \begin{aligned}
        &\sum\nolimits_{k=0}^{K-1}(\mathcal{L}_{\tau^{(k)}}^{(k)} - \mathcal{L}_{\tau^{(k+1)}}^{(k+1)}) \stackrel{}{=} \mathcal{L}_{\tau^{(0)}}^{(0)} - \mathcal{L}_{\tau^{(K)}}^{(K)}\\ 
        &\stackrel{(a)}{\geq} \sum_{k=0}^{K-1}[(\alpha/2)\|\Delta \mathbf{D}^{(k+1)}\|_{\rm F}^2
        + (\beta/2)\|\Delta \mathbf{P}^{(k+1)}\|_{\rm F}^2\\& - \delta \tau^{(k+1)}\|\Delta \mathbf{\Lambda}^{(k+1)}\|_{\rm F}^2],
    \end{aligned}
\end{equation}
where (a) is due to Lemma \eqref{lemma:1}.
Since $\mathbf{\Lambda}^{(k)}$ is bounded and $\mathbf{P}^{(k)}$ is feasible with respect to the power constraint for all $k$, we have $\lim_{K \to \infty} -\mathcal{L}_{\tau^{(K)}}^{(K)}< \infty$. 
Consequently, the right-hand sum in \eqref{eq:augineq} is also bounded as $K \to \infty$. 

Considering the series corresponding to the last term of \eqref{eq:augineq}, we obtain 
\begin{equation}
\label{eq:boundedlag}
\sum\nolimits_{k=0}^{\infty}\delta \tau^{(k+1)}\|\Delta \mathbf{\Lambda}^{(k+1)}\|_{\rm F}^2
\leq 2CM\sum\nolimits_{k=0}^{\infty}\frac{1}{\tau^{(k)}},
\end{equation}
by the boundedness of $\mathbf{\Lambda}^{(k)}$ and the bound $\tau^{(k+1)}/\tau^{(k)} \leq C$. Since $\frac{1}{\tau^{(k)}}$ is summable, we have that \eqref{eq:boundedlag} is finite. Hence, the contribution of $(\alpha/2)\|\Delta \mathbf{D}^{(k+1)}\|_{\rm F}^2$ and $(\beta/2)\|\Delta\mathbf{P}^{(k+1)}\|_{\rm F}^2$ is finite as $k \to \infty$ in \eqref{eq:augineq}; thus, $\mathbf{D}^{(k)} - \mathbf{D}^{(k+1)} \to 0, \mathbf{P}^{(k)} - \mathbf{P}^{(k+1)} \to 0$.

The next part of the proof is to show that the KKT-conditions corresponding to \eqref{eq:optprob}, 
\begin{equation}
\label{eq:kktconds}
\begin{aligned}
    &\text{KKT-1}: \mathbf{P}\mathbf{\Lambda} = \mathbf{D}\\
    &\text{KKT-2}: \mathbf{D}\mathbf{\Lambda}^{\rm T} = \mathbf{P}\text{diag}(\gamma_1,\ldots, \gamma_p)\\
    &\text{KKT-3}: (\mathbf{G}\odot \mathbf{H}^{-1})^{\rm T} = \mathbf{P}^{\rm T}\mathbf{D}\\
    &\text{KKT-4}: \gamma_j(\|\mathbf{p}_j\|_2^2-P_j) = 0, \: j=1,\ldots,p\\
    &\text{KKT-5}: \gamma_j \geq 0, \: j=1,\ldots p,\\
\end{aligned}
\end{equation}
are satisfied in the limit $k \to \infty$ by the update-rules in Algorithm \ref{alg:admm}.

Beginning with the update for $\mathbf{P}^{(k+1)}$ in Algorithm \ref{alg:admm}, we have that the optimality condition at iteration $k$ becomes
\begin{equation}
-\nabla_{\mathbf{P}} \hat{\mathcal{L}}_{\tau^{(k)}}(\mathbf{D}^{(k+1)},\mathbf{P}^{(k+1)},\mathbf{\Lambda}^{(k)}) \in N_{\Omega}(\mathbf{P}^{(k+1)}).
\end{equation}
Using the definition of the normal cone $N_{\Omega}$ of our set $\Omega$, and substituting $\mathbf{\Lambda}^{(k)} =\mathbf{\Lambda}^{(k+1)}-\tau^{(k)}((\mathbf{G} \odot \mathbf{H}^{-1})^{\rm T}-\mathbf{P}^{(k+1)^{\rm T}}\mathbf{D}^{(k+1)})$, we get
\begin{equation}
\label{eq:kktp}
\mathbf{D}^{(k+1)^{}}\mathbf{\Lambda}^{(k+1)^{\rm T}} = \mathbf{P}^{(k+1)}\mathbf{\Gamma}^{(k+1)} + \beta(\mathbf{P}^{(k+1)}-\mathbf{P}^{(k)})
\end{equation}
where $\mathbf{\Gamma}$ is a diagonal matrix of non-negative coefficients $\gamma_1^{(k+1)}, \ldots, \gamma_p^{(k+1)}$. Since we showed that $\mathbf{P}^{(k+1)} \to \mathbf{P}^{(k)}$ with $k \to \infty$, then \eqref{eq:kktp} approaches KKT-2 as $k \to \infty$, with $\mathbf{\Gamma}$ containing the coefficients $\gamma_1,\ldots,\gamma_p$ on the diagonal, satisfying KKT-4 and KKT-5 by the definition of the normal cone $N_{\Omega}(\mathbf{P}^{(k+1)})$ of the set $\Omega$ at $\mathbf{P}^{(k+1)}$.

The remaining updates are analogous to the derivations in \cite{MatFact}, and are therefore omitted in the interest of space.

\subsection{Proof of Lemma \ref{lem:2}}

Consider matrices $\mathbf{P}, \mathbf{D}, \mathbf{\Lambda}$ and scalars $\gamma_1,\ldots\gamma_p$ that satisfy the KKT conditions \eqref{eq:kktconds}.
We note that for $\mathbf{\Lambda}$ to become unbounded and still satisfy KKT-1 and KKT-2, we need 
\begin{equation}
    \label{eq:nulldelta}
    \mathbf{\Lambda} = t\mathbf{\Delta}+\mathbf{\Lambda}_0,
\end{equation}
with $\mathbf{\Lambda}_0$ satisfying KKT-1, KKT-2 and $\mathbf{\Delta}$ having columns in $N(\mathbf{P})$ and rows in $N(\mathbf{D})$, so that $t \to \infty$ in \eqref{eq:nulldelta} does not affect \eqref{eq:kktconds}. We have rank$(\mathbf{G}) = \min(m,p)$, so that $T \geq \min(m,p)$. Since $\mathbf{P}$ is $T \times s$, then the dimension of $N(\mathbf{P})$ is $\max(0,p-T)$, and similarly the dimension of $N(\mathbf{D})$ is $\max(0,m-T)$, and therefore $N(\mathbf{D}) \cap N(\mathbf{P}) = \{0\}$ so that $\mathbf{\Delta} = \mathbf{0}$.
\label{appendix:3}
\end{document}